\documentclass[journal,letterpaper]{IEEEtran}

\usepackage[utf8]{inputenc} 
\usepackage[T1]{fontenc}
\usepackage{url}
\usepackage{ifthen}
\usepackage{cite}
\usepackage[cmex10]{amsmath} 
                             
\usepackage[table]{xcolor}
\usepackage{amssymb,amsthm}
\usepackage{subcaption}
\usepackage{tikz}
\usepackage{hyperref}
\usepackage[normalem]{ulem}
\usepackage[noend]{algpseudocode}
\usepackage{algorithm}
\usepackage{tabularx,environ}
\usetikzlibrary{backgrounds,calc,shapes,arrows}

\DeclareMathOperator{\Span}{span}

\DeclareMathOperator{\diag}{diag}

\newtheorem{thm}{Theorem}
\newtheorem{cor}{Corollary}

\begin{document}
	
\title{From Priors to Projections: Geometry and simplified MIMO demodulation of probabilistic shaping}

 \author{%
	\IEEEauthorblockN{Kirill Ivanov, Wei Yang, Mahmoud Taherzadeh Boroujeni, Jing Jiang}\\
	\IEEEauthorblockA{Qualcomm Technologies, Inc., San Diego, CA, 92121\\
		Email: \{kivanov, weiyang, mtaherza, jingj\}@qti.qualcomm.com}
}

\maketitle

\begin{abstract}
	Probabilistic shaping (PS) is a well-known method to achieve improved performance upon a regular quadrature amplitude modulation (QAM) by taking the target constellation and making the distribution of underlying points non-uniform. It has been extensively studied over the years for the additive white Gaussian noise (AWGN) and Rayleigh fading channels. However, the potential of probabilistic shaping in the multiple-input and multiple-output (MIMO) setting needs further investigations.
	
	 In this paper, we prove that if shaped symbols follow Maxwell-Boltzmann distribution, the optimal maximum a posteriori (MAP) detection is equivalent to the case of uniform constellation with a simple preprocessing step. Our approach has multiple benefits. It enables to utilize the same processing chain for both shaped and unshaped system, which simplifies the receiver architecture. This technique can be applied for both linear MMSE and nonlinear (near-)MAP demapper types. In addition, the complexity of adaptive methods such as sphere decoding can be reduced.
\end{abstract}

\section{Introduction}
\IEEEPARstart{A}{s} the cellular communications evolve, every new generation is followed by the substantial increase in the peak data rate. The crucial components required for achieving high-throughput transmissions are the use of multiple transmit and receive antennas and the higher-order modulation. Typically, $2^{2M}$-QAM is utilized, where the constellation points are arranged in a quadratic grid, which allows to improve spectral efficiency with low demodulation complexity. However, the capacity-achieving distribution for channels with additive Gaussian noise is also Gaussian, and regular QAM constellations have an asymptotic 1.53dB gap to Shannon capacity in AWGN channels \cite{forney1989multidimensional}.

Constellation shaping is a method to bridge the capacity gap by modifying the constellation to better match the capacity-achieving distribution of the channel. Probabilistic shaping (PS) is a practical method that modifies the probabilities of the underlying constellation so that the channel inputs get closer to the optimal distribution. In particular, Maxwell-Boltzmann (MB) distribution is a practical choice for the wireless transmission channels \cite{kschischang1993optimal}. Various versions of PS have been proposed over the years, including schemes based on turbo \cite{raphaeli2004constellation} and LDPC \cite{kaimalettu2007constellation} codes.  Note that the receiver implementation can be kept relatively simple because the underlying constellation remains unchanged. On the other hand, geometric shaping (GS) rearranges the constellation points to target the similar performance gains. This method is utilized in ATSC 3.0 \cite{ATSC-A322} and DVB \cite{DVB-NGH2013} standards, albeit the lack of lattice structure in the constellation makes demodulation more complex.

A probabilistic shaping scheme proposed in \cite{bocherer2015bandwidth} is based on a distribution matcher (DM) that converts uniform data bits to modulation symbols with the target distribution together with a systematic error-correcting code and is shown to have good performance over a practical coded modulation system. PS (a.k.a., probabilistic amplitude shaping (PAS) in \cite{bocherer2015bandwidth}) is capable of achieving near-optimal shaping gains using constant composition distribution matcher (CCDM) \cite{schulte2016constant} and low-density parity-check (LDPC) codes at moderate to long block lengths even in case of the bit-interleaved coding modulation (BICM) \cite{caire1998bicm}, when the low-complexity receiver performs an independent demodulation of bits in the constellation symbol. Since then, there has been an extensive research on more efficient DM algorithms that have advantage over CCDM, especially for the small DM block length. Notable examples include multiset-partition distribution matching (MPDM)\cite{fehenberger2019multiset}, enumerative sphere shaping (ESS)\cite{gultekin2020enumerative} and energy-based shaping \cite{liu2023energy}. These methods typically have a larger storage or computational complexity cost but provide significant performance improvements in small to moderate block length regime, which makes them more suitable for the practical implementations.

Multi-antenna techniques already found their way to the modern telecommunication standards, such as 5G \cite{3gpp.38.211} and Wi-Fi \cite{IEEE80211n} networks. It has been long known that MIMO communication systems can benefit from the increased capacity and diversity via spatial multiplexing of several data streams. When the number of spatial layers used for transmission grows large, linear methods such as minimum mean square error (MMSE) or zero-forcing (ZF) become extremely suboptimal \cite{yang2015fifty}. Over the last several decades, many reduced-complexity maximum likelihood (ML) demappers have been proposed. The state-of-the-art solution is to utilize sphere decoder \cite{pohst1981computation,schnorr1991lattice,jalden2005complexity,studer2008soft} with sorted QR decomposition \cite{wubben2003mmse}. Additional complexity and performance gains can be obtained by lattice reduction \cite{yao2002lattice,taherzadeh2007lll, gan2009complex,gestner2010lattice}.

Probabilistic shaping is yet to be adopted in the modern telecommunication standards. So far the topic of PS in MIMO channels attracted little focus from academia, which was mostly focused on the case of AWGN channel and single-layer transmissions. Most publications on the topic use linear MMSE detector \cite{kang2022probabilistic,bobrov2023probability,hu2024supporting} and the demonstrated gains are less than 1dB. In our previous work \cite{ivanov2025probabilistic}, we demonstrated that PS can achieve performance gains larger than the 1.53dB AWGN upper bound with a nonlinear demapper such as sphere decoder and conjectured that this is due to the interference shaping.

An important consideration on the way to adoption of a new transmission scheme is the receiver complexity. In case of AWGN channel, soft-output demodulation can be easily performed with a simple rescaling \cite{jia2024simplified}, thus making the overall receiver complexity similar to uniform QAM. In this paper, we demonstrate that the same holds for the general case of MIMO channels. We investigate a MAP detection problem in the system with probabilistic shaping and prove that if the constellation symbols follow Maxwell-Boltzmann distribution, the detection problem can be reduced to the case of uniform distribution using a simple matrix preprocessing. Note that a similar result has been earlier derived in \cite{campello2018universal} in the context of lattice coding for MIMO. Our method fits naturally into QR-based detection methods and enables the unified receiver flow for both uniform and shaped constellation, which simplifies the implementation and makes probabilistic shaping a good candidate for the future telecommunication systems. The same formulation can also be used with linear methods such as MMSE. In addition, the proposed method reduces the complexity of channel-adaptive methods such as sphere decoder.

\section{Background}

\subsection{Probabilistic amplitude shaping}
Consider $2^M$-PAM constellation 
$$
\mathcal O_R=\{\pm 1,\pm 3,\dots,\pm (2^M-1)\}.
$$ 
It is well-known that for AWGN channel the capacity-achieving input distribution is Gaussian. For finite constellations such as $\mathcal O_R$, capacity can be approached if the points are selected for transmission according to the Maxwell-Boltzmann distribution \cite{kschischang1993optimal}
\begin{equation}
\label{eq:MB}
P_{\nu}(x)=Ze^{-\nu \|x\|^2}, x\in \mathcal O_R,
\end{equation}
where $Z$ is the normalization factor selected to ensure that $\sum_xP_{\nu}(x)=1$ and the parameter $\nu$ determines the source entropy.

Consider now an $(N, K)$ error-correcting code of rate $K/N=(M-1)/M$ with a systematic generator matrix
$$
\mathbf G=\begin{pmatrix}
	\mathbf I_k & \mathbf P
\end{pmatrix}.
$$
The encoding operation produces a length-$N$ codeword $\mathbf c=\begin{pmatrix}
	\mathbf d_K & \mathbf p_{N-K}
\end{pmatrix}$, where the first $K$ bits are equal to the data bits and hence have an identical distribution, whereas the remaining $N-K$ bits can be considered uniform \cite{bocherer2015bandwidth}. Since $P_{\nu}(x)=P_{\nu}(-x)$, we can map the data bits to the amplitudes of constellation symbols and the parity bits to the signs, i.e. $\mathbf d_K$ is decomposed into $K/(M-1)$ groups $A_i$ each containing $M-1$ bits, $\mathbf p_{N-K}$ is decomposed into $K/(M-1)$ singletons $P_i$ and we map the pair $(A_i,P_i)$ to the elements of $\mathcal O_R$ using Gray labeling. Note that the generalization to $2^{2M}$-QAM, defined as the direct product of two $2^{M}$-PAM constellations, is straightforward. In case when the code rate is greater than $(M-1)/M$, some information bits are left unshaped and mapped to the sign bits. The only remaining ingredient is how to generate the vector $\mathbf d_K$ so that the empiric distribution of amplitudes $A_i$ is close to \eqref{eq:MB}.

In this paper, we consider constant composition distribution matcher (CCDM) \cite{schulte2016constant}. It takes $K_c$ uniform data bits as an input and outputs $N_c$ amplitudes $A_i$ that follow the distribution $P_{\nu}$. For an amplitude sequence $A^{N_c}$, define $\mathcal T_x^{N_c}$ as the set of all permutations of $A^{N_c}$. The elements of $\mathcal T_x^{N_c}$ have the same empirical distribution $P(x)=\frac{|i: A_i=x|}{N_c}$, which leads to the following procedure\cite{bocherer2015bandwidth}:
\begin{enumerate}
	\item Compute $N_x=\left[N_c\cdot P_{\nu}(x)\right], x\in \{1,3,\dots,(2^M-1)\}$
	\item Select a sequence $A^{N_c}$ with an empirical distribution $P(x)=N_x/N_c$ and construct the set $\mathcal T_x^{N_c}$
	\item Compute the input length 
	$$
	K_c=\left\lfloor\log_2\frac{N_c!}{N_1!N_3!\dots N_{2^M-1}!}\right\rfloor
	$$
	\item Pick the set $\mathcal C_{CCDM}$ of $2^{K_c}$ sequences from $\mathcal T_x^{N_c}$ and define a bijective mapping $\{0,1\}^{K_c}\to \mathcal C_{CCDM}$.
\end{enumerate}
The mapping proposed in \cite{schulte2016constant} is based on the arithmetic coding and generalizes the idea from \cite{ramabadran1990coding} to the non-binary case. Note that the ideas presented in the main section apply to any DM that outputs symbols with distribution close to MB, e.g. MPDM or energy-based shaping methods.

\subsection{MIMO detection}
In this and the next sections, we mostly follow the notations from \cite{studer2008soft} that are adjusted to account for potentially non-uniform constellation symbol probabilities. 

Consider the complex-valued set $\mathcal O$ of constellation points. 
A $M_t\times M_r$ MIMO system with $M_t$ transmit antennas and $M_r$ receive antennas can be characterized using the input-output relation 
\begin{equation}
\label{eq:mimo}
\mathbf y=\mathbf H\mathbf s + \mathbf n,
\end{equation}
where $\mathbf H\in \mathbb C^{M_t\times M_r}, \mathbf n\sim\mathcal{C}\mathcal N(0,\sigma \mathbf I_{M_r})$ and $\mathbf s\in \mathcal O^{M_t}$. For simplicity, in what follows we assume the unit noise power. 

Assume that the receiver has perfect knowledge of the matrix $\mathbf H$. The maximum a posteriori (MAP) estimate of the transmitted symbol sequence in system \eqref{eq:mimo} is given by 
\begin{align}
	\label{eq:map}
	\mathbf s^{MAP}&=\arg \max_{\mathbf s \in  \mathcal O^{M_t}}P(\mathbf s | \mathbf y)\\
	&=\arg \min_{\mathbf s \in  \mathcal O^{M_t}}\left(\|\mathbf y-\mathbf H\mathbf s\|^2 - \log P(\mathbf s)\right).
\end{align}

In practical MIMO systems, coded bits are mapped to the constellation symbols by assigning distinct $B$-bit labels to each constellation symbol. Let $s_{j,b}$ denote $b$-th label bit in the transmitted symbol $s_j$. In a coded MIMO system, soft-output demapper needs to compute the log-likelihood ratios (LLRs) 

$$
L_{j,b}=\log \frac{P(s_{j,b}=0|\mathbf y)}{P(s_{j,b}=1|\mathbf y)},0\le b < B, 0\le j < M_t,
$$

which are subsequently passed to the channel decoder. Max-log-MAP LLRs can be computed as
\begin{align}
	\label{eq:mlm}
	L_{j,b}(\mathbf y)=&\min_{\mathbf s\in \mathcal{X}_{j,b}^{(0)}}\left( \|\mathbf y-\mathbf H\mathbf s\|^2 - \log P(\mathbf s)\right)\\-&\min_{\mathbf s\in \mathcal{X}_{j,b}^{(1)}}\left( \|\mathbf y-\mathbf H\mathbf s\|^2 - \log P(\mathbf s)\right)\notag ,
\end{align}
where $\mathcal{X}_{j,b}^{(i)}$ are the subsets of $\mathcal O^{M_t}$ that correspond to $s_{j,b}=i$.

\subsection{Sphere decoder}
\label{ss:sd}

MIMO detection problem can be transformed into a tree search problem using the QR decomposition of the channel matrix $\mathbf H $ \cite{pohst1981computation,schnorr1991lattice,jalden2005complexity}. We obtain $\mathbf H=\mathbf Q \mathbf R$, where $M_r\times M_t$ matrix $\mathbf Q$ is unitary and $M_t\times M_t$ upper-triangular matrix $\mathbf R$ has real-valued positive entries on its diagonal. Multiplying both sides of \eqref{eq:mimo} by $\mathbf Q^H$ leads to the alternative formulation 
\begin{align}
	\label{eq:mlm2}
L_{j,b}(\mathbf{\tilde y})=&\min_{\mathbf s\in \mathcal{X}_{j,b}^{(0)}}\left( \|\mathbf{\tilde y}-\mathbf R\mathbf s\|^2 - \log P(\mathbf s)\right)\\-&\min_{\mathbf s\in \mathcal{X}_{j,b}^{(1)}}\left( \|\mathbf{\tilde y}-\mathbf R\mathbf s\|^2 - \log P(\mathbf s)\right)\notag ,
\end{align}
where $\mathbf{\tilde y}=\mathbf Q^H\mathbf y$. Using the upper-triangular structure of matrix $\mathbf R$, for $i=M_t-1,\dots,0$ we define the partial distances
\begin{align}
	\label{eq:ped}
	d_i(\mathbf s)&=d_{i+1}(\mathbf s)+e_i(\mathbf s)-\log P(s_j|s_{j+1},\dots,s_{M_t-1}),\\
	e_i(\mathbf s)&=\left\|\tilde y_i-\sum_{j=i}^{M_t-1}R_{i,j}s_j\right\|^2\notag,
\end{align}
that can be computed recursively. Note that while the efficient enumeration of constellation points is in general a highly nontrivial task due to the term $\log P(s_j|\cdot)$, in case of QAM constellation with uniformly distributed points its lattice structure enables low-complexity closest-point search, as well as the further enumeration \cite{pohst1981computation, schnorr1991lattice}.

Now consider a tree with $M_t$ layers, where the nodes are elements of $\mathcal O$, the branches correspond to the values $d_i$ and each path from the root to a leaf corresponds to a symbol vector $\mathbf s\in\mathcal O^{M_t}$. The MAP solution $\mathbf s_{MAP}$ is the path with the smallest metric $d_0$. We start from the infinite search radius $r$ and traverse the tree depth-first, updating $r$ whenever a leaf is reached and pruning the branches with $d_i>r$.

The generation of LLRs using \eqref{eq:mlm2} can be decomposed into two steps. The first step is to find the MAP solution 
\begin{equation}
\label{eq:sdmap}
\mathbf s^{MAP}=\arg \min_{\mathbf s \in  \mathcal O^{M_t}}\left(\|\mathbf{\tilde y}-\mathbf R\mathbf s\|^2 - \log P(\mathbf s)\right).
\end{equation}
Assume that $s^{MAP}_{j,b}=i$. The second step is to find the best counter-hypothesis
\begin{equation}
\label{eq:counter}
\mathbf s^{\overline{MAP}}_{j,b}=\arg \min_{\mathbf s\in \mathcal{X}_{j,b}^{(\overline i)}} \left(\|\mathbf{\tilde y}-\mathbf R\mathbf s\|^2 - \log P(\mathbf s)\right),
\end{equation}
i.e. to find the most likely transmitted vector $\mathbf s\in\mathcal O^{M_t}$ s.t. $s_{j,b}\neq i$. We use the repeated tree search (RTS) strategy with modified detection order\cite{wang2004approaching,marsh2005smart} and run the tree search $2M\cdot M_t + 1$ times:
\begin{enumerate}
	\item During the first run, we find the MAP solution $\mathbf s^{MAP}$ and store it.
	\item For each symbol $s_j,j=0,\dots,M_t-1$, we rearrange the columns of $\mathbf H$ so that the detection process starts from layer $j$. Then for each bit $b$ the search starts from the reduced candidate set $\mathcal{X}_{j,b}^{(\overline i)}$ at the first layer and uses the full candidate set $\mathcal O$ for all subsequent layers, which makes the tree pruning efficient and guarantees that only valid counter-hypotheses are considered.
\end{enumerate}
We sort the columns of $\mathbf H$ before each QR decomposition so that the stronger layers are closer to the tree root \cite{wubben2003mmse}.

Another substantial complexity reduction is achieved by LLR clipping \cite{studer2008soft}. Since $\left|L_{j,b}\right|=\left|d(\mathbf s^{MAP})-d(\mathbf s^{\overline{MAP}}_{j,b})\right|$, after finding $\mathbf s^{MAP}$ we can simply initialize the search radius for the subsequent runs as $$r\gets d(\mathbf s^{MAP})+\lambda,$$ where $\lambda$ is the clipping value.

\subsection{MMSE detection}
\label{ss:mmse}
A minimum mean square error (MMSE) estimate is given by 
$$
\mathbf s^*=\arg \min_{\mathbf s} \mathbb{E}[\|\mathbf y-\mathbf H\mathbf s\|^2].
$$
A linear MMSE estimate $\hat{\mathbf s}=\mathbf W\mathbf y$ can be obtained by applying the filter matrix $\mathbf W\in \mathbb C^{M_t\times M_r}$ which minimizes $\mathbb E [\|\mathbf W\mathbf y-\mathbf s \|^2]$ and can be computed as 
\begin{equation}
	\label{eq:mmse}
	\mathbf W = \left(\mathbf H^H\mathbf H + \mathbf I\right)^{-1}\mathbf H^H.
\end{equation}

When $\mathbf s$ and $\mathbf y$ are jointly Gaussian, the two estimates coincide, i.e. $\mathbf s^*=\hat{\mathbf s}$. However, this property no longer holds for the case of discrete $\mathbf s$, except for the low-SNR regime. Despite that, linear MMSE receiver is widely used in the real-world MIMO systems due to its simplicity and has been extensively studied in academia \cite{kim2008performance, mckay2010achievable, yang2015fifty}. The corresponding LLR values can be obtained as

\begin{align}
	\label{eq:mmsemlm}
	L_{j,b}(\mathbf y)=&\min_{s_j\in \mathcal{X}_{b}^{(0)}}\left( \|\hat{s}_j-s_j\|^2 - \log P(s_j)\right)\\-&\min_{s_j\in \mathcal{X}_{b}^{(1)}}\left( \|\hat{s}_j-s_j\|^2 - \log P(s_j)\right)\notag ,
\end{align}
where $\mathcal{X}_{b}^{(i)}$ play similar role as $\mathcal{X}_{j,b}^{(i)}$ in \eqref{eq:mlm} but only contain the elements of currently demapped layer.

\section{Simplified PS demapper}

\subsection{Shaping gain over uniform QAM}
Assume that the set $\mathcal O$ corresponds to $2^{2M}$-QAM constellation. By changing the symbol distribution from uniform to Maxwell-Boltzmann, the transmitter obtains a power gain
$$
\mathcal P_{\nu}=\frac{\mathbb E_{P_{u}}[\|\mathbf s\|^2]}{\mathbb  E_{P_{\nu}}[\|\mathbf s\|^2]},
$$
where $\mathbb E_{P}$ denotes the expectation wrt distribution $P$ and $P_u$ is the uniform distribution. Hence, assuming the formulation \eqref{eq:mimo} for the unshaped MIMO setup, its shaped counterpart can be written as 

\begin{equation}
	\label{eq:shapedmimo}
	\mathbf y=\tilde{\mathbf H}\mathbf s + \mathbf n,
\end{equation}
where $\tilde{\mathbf H}=\mathbf H \mathbf D$ and $\mathbf D = \diag(\sqrt{\mathcal P_{\nu_0}},\dots,\sqrt{\mathcal P_{\nu_{M_t-1}}})$. With a slight abuse of notation, we assume $\mathbf H=\mathbf H \mathbf D$ for the rest of the paper.

\subsection{Simplified demapper for probabilistic shaping}

\label{s:mimopas}
Consider now the MIMO system where we apply probabilistic shaping independently at each layer according to MB distribution \eqref{eq:MB}, i.e. $P(s_0,\dots,s_{M_t-1})=\prod_jP(s_j)$ and $P(s_j)=P_{\nu_j}(s_j)$, and define $\mathbf V=\diag(\nu_0,\dots,\nu_{M_t-1})$.

\begin{thm}[PS detection as augmented QR]
	\label{thm:projmap}
	For a shaped MIMO system with MB prior, its MAP estimate is given by
	\begin{equation}
		\label{eq:projmap}
		\mathbf s^{MAP}=\arg \min_{\mathbf s \in  \mathcal O^{M_t}}\|\mathbf Q_{V}^H \mathbf y-\mathbf R_{ V}\mathbf s\|^2,
	\end{equation}
	where $\mathbf R_{V}^H\mathbf R_{V}=\mathbf H^H\mathbf H+\mathbf V$ and $\mathbf Q_{V}\mathbf R_{V}=\mathbf H$.
\end{thm}

\begin{proof}
Take the right hand side of \eqref{eq:map} and plug in $P(s_j)\propto e^{-\nu_j \|s_j\|^2}$. We get 
$$
	\mathbf s^{MAP}=\arg \min_{\mathbf s} \left(\|\mathbf y-\mathbf H\mathbf s\|^2+\mathbf s^H\mathbf V\mathbf s\right).
$$
Expanding the $\|\|^2$ gives 
\begin{align*}
&\|\mathbf y-\mathbf H\mathbf s\|^2+\mathbf s^H\mathbf V\mathbf s\\
&=\|\mathbf y\|^2-\mathbf y^H\mathbf H\mathbf s-\mathbf s^H\mathbf H^H\mathbf y+\mathbf s^H\mathbf H^H\mathbf H\mathbf s+\mathbf s^H\mathbf V\mathbf s\\
&=\|\mathbf y\|^2-\mathbf y^H\mathbf H\mathbf s-\mathbf s^H\mathbf H^H\mathbf y+\mathbf s^H\left(\mathbf H^H\mathbf H+\mathbf V\right)\mathbf s\\
&=\|\mathbf y\|^2-\mathbf y^H\mathbf Q_V\mathbf R_V\mathbf s-\mathbf s^H\mathbf R_V^H\mathbf Q_V^H\mathbf y+\mathbf s^H\mathbf R_{V}^H\mathbf R_{V}\mathbf s\\
&=\|\mathbf Q_V^H\mathbf y-\mathbf R_V\mathbf s\|^2-\|\mathbf Q_V^H\mathbf y\|^2+\|\mathbf y\|^2\notag\\
&=\|\mathbf Q_V^H\mathbf y-\mathbf R_V\mathbf s\|^2+f(\mathbf y).
\end{align*}
It remains to notice that $f(\mathbf y)$ is independent of $\mathbf s$ and therefore 
\begin{align*}
	\mathbf s^{MAP}&=\arg \min_{\mathbf s} \left(\|\mathbf Q_V^H\mathbf y-\mathbf R_V\mathbf s\|^2+f(\mathbf y)\right)\\
	&=\arg \min_{\mathbf s}\left(\|\mathbf Q_V^H\mathbf y-\mathbf R_V\mathbf s\|^2\right).
\end{align*}
\end{proof}

Our result shows that instead of directly embedding symbol probabilities $P(s_j)$ into the detection algorithm, the receiver can instead perform preprocessing of the channel matrix by computing 
QR decomposition of the augmented matrix $\mathbf H'=\begin{pmatrix}
	\mathbf H \\ \mathbf V^{1/2}
\end{pmatrix}$ as $\text{QR}(\mathbf H')=\begin{pmatrix} \mathbf Q_0 \\ \mathbf Q_1\end{pmatrix}\mathbf R$ and assume the uniform probability distribution subsequently. It follows that $\mathbf Q_V=\mathbf Q_0$ and $\mathbf R_V=\mathbf R$.

Note that this formulation resembles the commonly used augmented QR method for MIMO detection, where QR decomposition is applied to the matrix $\tilde{\mathbf H}=\begin{pmatrix}
\mathbf H \\ \mathbf I
\end{pmatrix}$ \cite{wubben2003mmse}. Combining this with our approach, we get the augmented channel matrix $\tilde{\mathbf H}_V=\begin{pmatrix}
\mathbf H \\ (\mathbf I + \mathbf V)^{1/2}
\end{pmatrix}$. Observe that for the case of AWGN transmission, the proposed approach reduces to the method described in \cite{jia2024simplified}.

\begin{cor}
For a shaped MIMO system with MB prior, its linear MMSE estimate can be computed as 
\begin{equation}
	\label{eq:projmmse}
	\hat{\mathbf s}=\mathbf W_V\mathbf y,
\end{equation}
where 
$\mathbf W_V=\left(\mathbf H^H\mathbf H+ \mathbf V + \mathbf I\right)^{-1}\mathbf H^H$. 
\end{cor}
\begin{proof}
Replace $\mathbf y$ and $\mathbf H$ in \eqref{eq:mmse} with $\mathbf Q_V^H\mathbf y$ and $\mathbf R_V$ to get 
\begin{align*}
\hat{\mathbf s}=&\left(\mathbf R_V^H\mathbf R_V+\mathbf I\right)^{-1}\mathbf R_V^H\mathbf Q_V^H\mathbf y\\
=&\left(\mathbf H^H\mathbf H+ \mathbf V + \mathbf I\right)^{-1}\mathbf H^H\mathbf y=\mathbf W_V\mathbf y.
\end{align*}
\end{proof}

The proposed framework incorporates the prior information into the filter matrix $\mathbf W_V$, so the corresponding LLRs can be obtained as 
\begin{align}
	\label{eq:mmsepsmlm}
	L_{j,b}(\mathbf y)=&\min_{s_j\in \mathcal{X}_{b}^{(0)}}\left( \|\hat{s}_j-s_j\|^2 \right)\\-&\min_{s_j\in \mathcal{X}_{b}^{(1)}}\left( \|\hat{s}_j-s_j\|^2\right)\notag.
\end{align}
Note that the equations \eqref{eq:mmsepsmlm} and \eqref{eq:mmsemlm} are not equivalent and the LLRs obtained through them are not guaranteed to be equal. However, in practice the generated LLRs are very close and both methods demonstrate comparable performance. As an example, Figure \ref{fig:mmse_pdf} demonstrates the empiric LLR distribution for the shaped amplitude bit in $4\times4$ MIMO scheme assuming shaped 16QAM with $\nu=0.303953$. We assume that the transmission takes place through TDL-A \cite{3gpp.38.901} channel with 30ns delay spread and 11 Hz Doppler frequency at 10dB SNR. We can observe that while the two distributions are not identical, they are very close to each other.
\begin{figure}
	\centering
	\includegraphics[width=\linewidth]{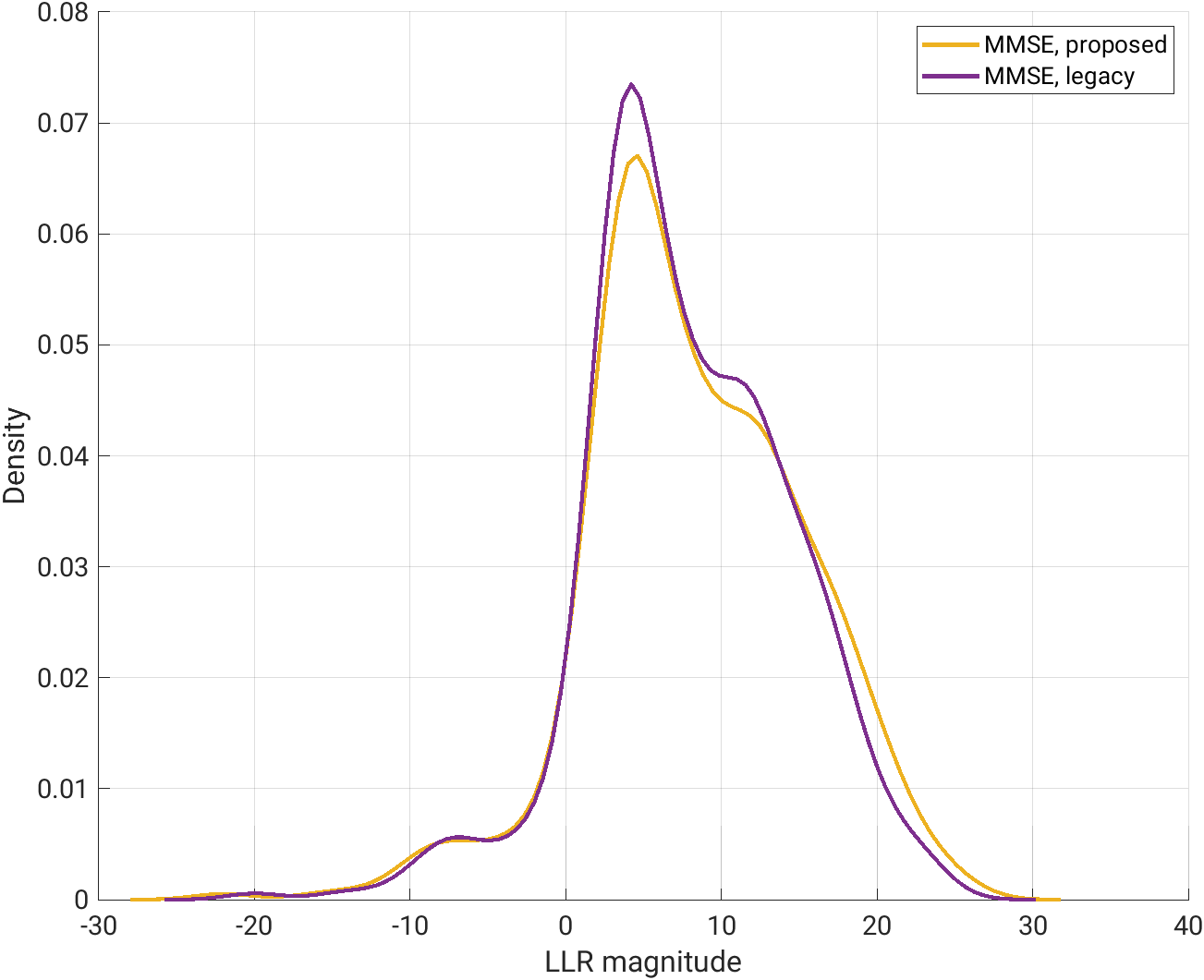}
	\caption{Empiric LLR distribution}
	\label{fig:mmse_pdf}
\end{figure}

The main advantage of the proposed method is the simplified receiver flow. It lets the demapping complexity of shaped constellation become essentially identical to the case of uniform a priori symbol distribution, which in case of QAM enables to fully leverage the underlying lattice structure. It can be considered as a preprocessing step before invocation of the main body of the detector algorithm and hence applies to any receiver type, from linear MMSE to the nonlinear methods such as sphere decoder. The proposed method provably works only for the case of Maxwell-Boltzmann prior and in general is not compatible with an arbitrary distribution. However, our experiments showed that it also works well when the distribution is sufficiently Gaussian-like, which is a good practical choice for the channels with additive Gaussian noise.

\subsection{Geometric interpretation}

The result of Theorem~\ref{thm:projmap} admits an intuitive geometric picture: the prior information $-\log P(\mathbf s)$ can be treated as if it were produced by extra, noiseless observation dimensions appended to the channel, so that MAP detection for probabilistic shaping reduces to an orthogonal projection onto the subspace spanned by this augmented channel, followed by ordinary MIMO (soft) detection.

Consider first a single real PAM layer with amplitude $s\in \mathcal O_R$ and unit channel gain, for which the MAP metric is $(y-s)^2+\nu s^2$. 
This MAP metric can be viewed as a squared Euclidean distance in an auxiliary two-dimensional plane between received signal $y'=(y,0)$ and lifted constellations $\{(s,\sqrt\nu s)$, as shown in Fig.~\ref{fig:geom1d}. 

The key observation is that the lifted constellations $\{(s,\sqrt\nu s):s\in\mathcal O_R\}$ lie entirely on a single line through the origin with slope $\sqrt\nu$ -- precisely the column space of the augmented channel $h'=(1,\sqrt\nu)^T$ underlying \eqref{eq:projmap}. 
In order to compute the distance between $y'$ and $(s,\sqrt\nu s)$, it suffices to project $y'$ onto the line (i.e., 1D column space) spanned by $h'$, and measure the distances therein. The distance between $y'$ and the projected point are common across all the constellation points, and hence can be ignored from the MAP metric. 
Furthermore, the lifted constellations remain to form a uniformly spaced lattice in this 1D space, with local coordinates ${r_\nu s}$, where $r_\nu=\sqrt{1+\nu}$.  
Finally, writing $q=h'/r_\nu$ for the unit vector spanning this line, the scalar coordinate of the projected point along the line can be written as $p=qq^Hy'=y/r_\nu$, where $qq^H$ is the rank-one projector onto $\Span(h')$. As we can see, $r_\nu s$ and $p=y/r_\nu$ corresponds to the two terms on the RHS of~\eqref{eq:projmap}, respectively. This proves Theorem~\ref{thm:projmap} in the scalar case.

\begin{figure}
	\centering
	\begin{tikzpicture}[scale=0.55]
		\draw[gray,->] (-8.5,0) -- (8.7,0);
		\foreach \x/\lbl in {-7/{-7},-5/{-5},-3/{-3},-1/{-1},1/{1},3/{3},5/{5},7/{7}}{
			\filldraw[blue] (\x,0) circle (2.2pt);
			\node[below] at (\x,-0.3) {\scriptsize \lbl};
		}
		\draw[thick,gray] (-8,-4) -- (8,4);
		\foreach \x in {-7,-5,-3,-1,1,3,5,7}{
			\draw[gray!60,thin] (\x,0) -- (\x,{0.5*\x});
		}
		\draw[gray!70,dashed] (-3.5,0) -- (-5,0) -- (-5,-2.5);
		\node[above] at (-4.25,-0.15) {\scriptsize $|y-s|$};
		\node[left] at (-5.15,-1.3) {\scriptsize $\sqrt\nu|s|$};
		\filldraw[green!60!black] (-5,-2.5) circle (2.8pt);
		\node[below] at (-5,-2.9) {\scriptsize $(s,\sqrt\nu s)$};
		\filldraw[red] (-3.5,0) circle (2.8pt);
		\node[above] at (-3.5,0.2) {\scriptsize $y'=(y,0)$};
		\filldraw[orange!80!black] (-2.8,-1.4) circle (2.6pt);
		\node[below right] at (-2.8,-1.4) {\scriptsize $p=qq^Hy'$};
		\draw[green!60!black, very thick] (-3.5,0) -- (-2.8,-1.4);
		\draw[black, thin] (-2.9342,-1.1317) -- (-2.6659,-0.9975) -- (-2.5317,-1.2658);
		\draw[orange!80!black, dashed] (-2.8,-1.4) -- (-5,-2.5);
	\end{tikzpicture}
	\caption{Geometric picture of the SISO MAP metric $(y-s)^2+\nu s^2$}
	\label{fig:geom1d}
\end{figure}
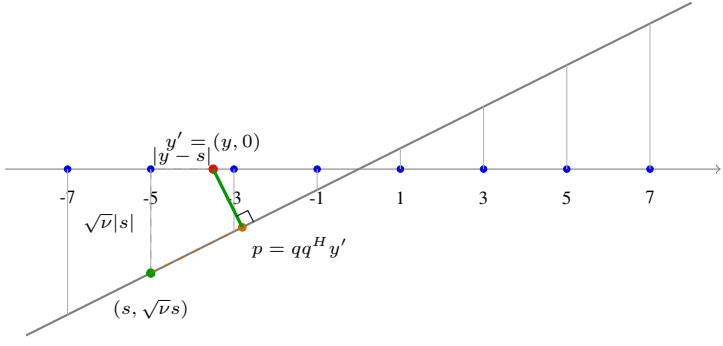

This picture generalizes directly to the multi-layer MIMO case. Each transmit layer contributes one signal dimension through the corresponding column of $\mathbf H$ and one shaping dimension through the corresponding column of $\mathbf V^{1/2}$, so the noiseless augmented received vector $\mathbf H'\mathbf s=\begin{pmatrix}\mathbf H \mathbf s \\ \mathbf V^{1/2}\mathbf s\end{pmatrix}$ traces out a lattice that lives in the subspace spanned by the columns of $\mathbf H'$. 
With the QR decomposition $\text{QR}(\mathbf H')=\begin{pmatrix}\mathbf Q_0\\ \mathbf Q_1\end{pmatrix}\mathbf R$, the matrix $\begin{pmatrix}\mathbf Q_0\\ \mathbf Q_1\end{pmatrix}$ provides an orthonormal basis to this subspace, while $\mathbf R=\mathbf R_V$ generates the lattice points with respect to this basis. 
Left-multiplying the observation $\mathbf y$ by $\mathbf Q_0^H=\mathbf Q_V^H$ is precisely the projection step, producing the sufficient statistic $\mathbf Q_V^H\mathbf y$ used in \eqref{eq:projmap}. Consequently, the receiver does not need to evaluate the log-prior term explicitly: detection is carried out against a lattice that has already been reshaped by the QR step to account for the MB prior, which is why a demodulator designed for uniform QAM can be reused verbatim, with $\mathbf Q_V^H\mathbf y$ and $\mathbf R_V$ simply taking the place of $\mathbf y$ and $\mathbf H$.

\section{Results}
We evaluate the performance of $4\times4$ MIMO system with 5G LDPC code \footnote{Strictly speaking, 5G LDPC code is not systematic, since a subset of the systematic bits are punctured from the codeword \cite{richardson2018design}. In our evaluations, we apply CCDM only on the set of systematic bits that are transmitted (i.e., not punctured) in order to preserve the shaping.} \cite{3gpp.38.212}. We use layered min-sum belief propagation decoder \cite{hocevar2004reduced} with 25 iterations for decoding of LDPC code. Table \ref{tbl:sim} provides the simulation parameters.

\begin{table}[]
	\caption{Simulation assumptions}
	\label{tbl:sim}
		\begin{tabular}{|l|l|}
			\hline
			Channel model            & \begin{tabular}[c]{@{}l@{}}TDL-A\cite{3gpp.38.901}\\ 30ns delay spread\\ 11Hz Doppler frequency\end{tabular}            \\ \hline
			Antennas          & 4Tx, 4Rx                                                                                                 \\ \hline
			OFDM symbols          & 1                                                                         \\ \hline
			Resource blocks          & 24                                                                                               \\ \hline
			Information bits  & 5946                                                                                                     \\ \hline
			Block length        & 9216                                                                                                     \\ \hline
			Coding rate         & 0.85                         \\ \hline
			Shaping        & CCDM with $\nu=0.025$                           \\ \hline
			Modulation          & 256QAM                                                                                                   \\ \hline
		\end{tabular}
	\centering
\end{table}

Figure \ref{fig:sd_perf} demonstrates the PS performance with sphere decoder, comparing the method from Section \ref{s:mimopas} to the state-of-the-art approach utilizing prior symbol probabilities in the tree search, as described in Section \ref{ss:sd}. As expected, there is no performance difference between the two methods. On the other hand, Figure \ref{fig:sd_compl} demonstrates the average number of nodes visited by the sphere decoder in each case. It can be seen that the proposed method brings a substantial reduction in complexity, bringing the average number of nodes down by 3x in the low-SNR region and by 20\% in the high-SNR region.  The observation is consistent with the geometric framework from the previous section: our method performs a change of basis of the channel matrix $\mathbf H$ and after accounting for shaping the new basis becomes more orthogonal, meaning that the interference is suppressed more efficiently and the tree search needs to explore fewer candidates at each level. The advantage is especially notable for the low-SNR region, whereas for high SNR the number of good candidates already drops down substantially, so the room for improvement is much smaller. 

It is worth noting that our method provides stable complexity across the considered SNR range, which is a stark contrast to the legacy sphere decoder behavior, whose visited-node count grows sharply at low SNR, as demonstrated on Figure \ref{fig:sd_compl}. This result demonstrates the efficiency of our approach and suggests a performance benefit for the case of complexity-constrained detectors, thus making it significantly more suitable for a practical implementation.

\begin{figure}
	\centering
	\includegraphics[width=\linewidth]{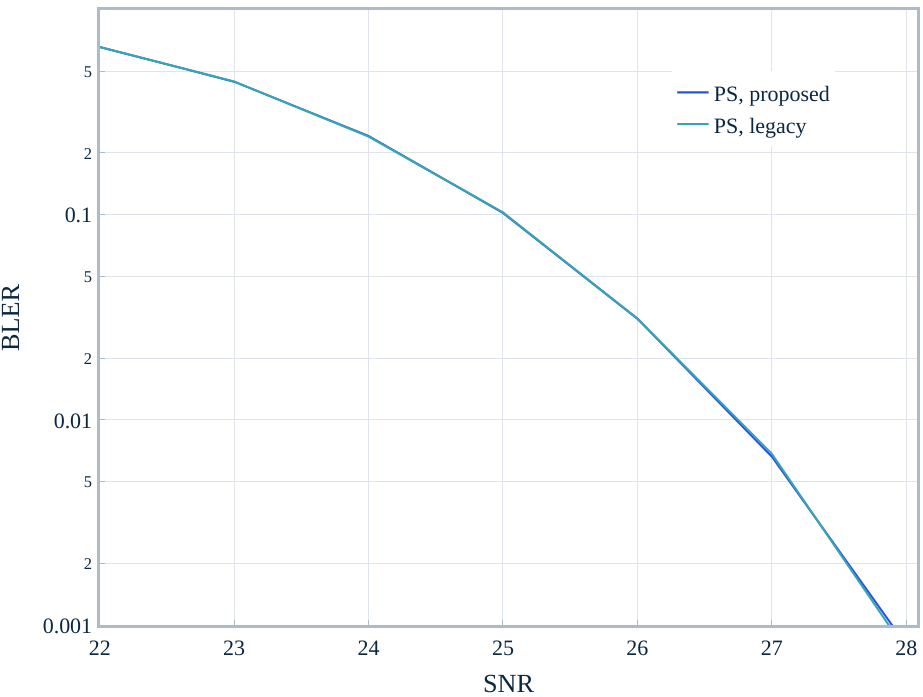}
	\caption{Sphere decoder performance}
	\label{fig:sd_perf}
\end{figure}

\begin{figure}
	\centering
	\includegraphics[width=\linewidth]{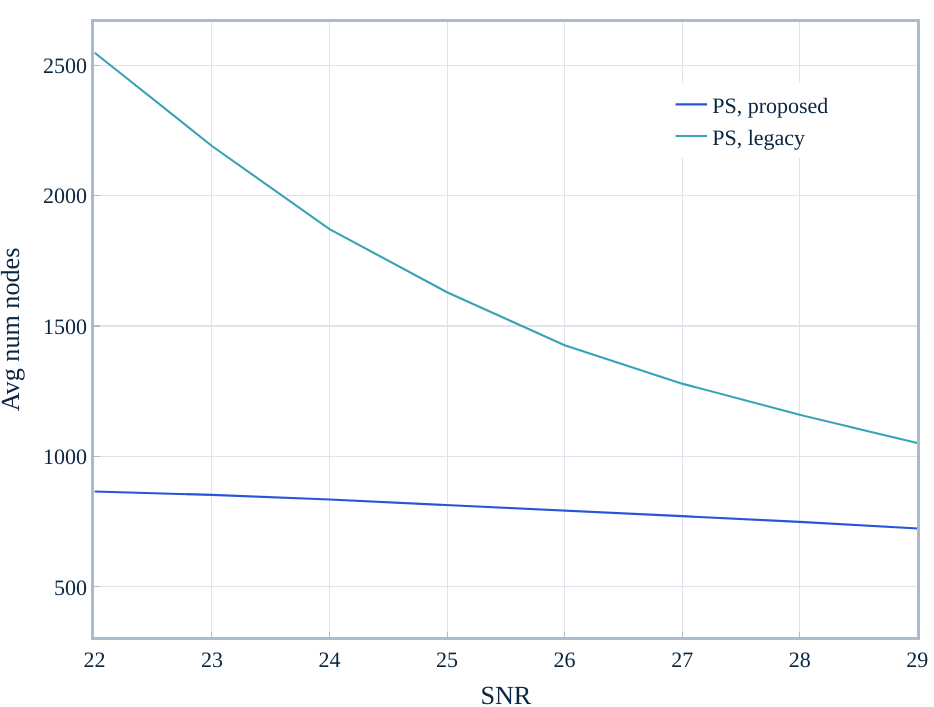}
	\caption{Sphere decoding complexity}
	\label{fig:sd_compl}
\end{figure}

\section{Conclusion}
In this paper, we investigated the receiver design for probabilistic shaping in MIMO channels. Our main result is the new formulation for the MAP detection, which embeds the  prior distribution of constellation symbols into the channel matrix used by the demapper algorithm. This way, after the preprocessing step the main receiver flow remains unchanged, which is beneficial from the practical implementation standpoint. Our result is very timely because the nonlinear detector is crucial for leveraging the interference shaping and achieving the largest gains in coded MIMO systems. Moreover, the new formulation reduces the complexity of adaptive algorithms, which also suggests that it can bring an additional performance gain when paired with a fixed-complexity demapper.

Our framework enables to utilize the same approach for linear MMSE receiver. In this case, the estimate is not exactly equivalent to the state-of-the-art approach but the generated LLRs are very close and consequently the performance remains unchanged. Hence, it enables the unified receiver architecture, which does not depend on the demapper type and remains similar for the case of uniform and shaped QAM constellation.

Currently, theoretical analysis of shaping performance and its limitations in MIMO fading channels has not received significant attention in academia. We hope that the presented formulation opens up the path for a deeper dive into this topic and turns it into a fruitful research direction. Potential extensions of our work include generalization to other symbol distributions, as well as to transmission setups that go beyond typical wireless scenarios.

\enlargethispage{-1.2cm} 

\bibliographystyle{IEEEtran}
\bibliography{biblio}

@techreport{3gpp.38.212,
	author = {3GPP},
	month = {01},
	note = {Version 15.4.0},
	number = {38.212},
	title = {{NR; Multiplexing and channel coding}},
	type = {Technical Specification (TS)},
	year = {2019}
}

@techreport{3gpp.38.211,
	author = {3GPP},
	month = {07},
	note = {Version 15.2.0},
	number = {38.211},
	title = {{NR; Physical channels and modulation}},
	type = {Technical Specification (TS)},
	year = {2018}
}

@techreport{3gpp.38.901,
	author = {3GPP},
	month = {11},
	note = {Version 16.1.0},
	number = {38.901},
	title = {{5G; Study on channel model for frequencies from 0.5 to 100 GHz}},
	type = {Technical Specification (TS)},
	year = {2020}
}

@ARTICLE{caire1998bicm,
	author={Caire, G. and Taricco, G. and Biglieri, E.},
	journal={IEEE Transactions on Information Theory}, 
	title={Bit-interleaved coded modulation}, 
	year={1998},
	volume={44},
	number={3},
	pages={927-946},
	doi={10.1109/18.669123}
}

@ARTICLE{forney1989multidimensional,
	author={Forney, G.D. and Wei, L.-F.},
	journal={IEEE Journal on Selected Areas in Communications}, 
	title={Multidimensional constellations. {I. I}ntroduction, figures of merit, and generalized cross constellations}, 
	year={1989},
	volume={7},
	number={6},
	pages={877-892},
	doi={10.1109/49.29611}
}

@article{kschischang1993optimal,
	title={Optimal nonuniform signaling for {G}aussian channels},
	author={Kschischang, Frank R and Pasupathy, Subbarayan},
	journal={IEEE Transactions on Information Theory},
	volume={39},
	number={3},
	pages={913--929},
	year={1993},
	publisher={IEEE}
}

@ARTICLE{bocherer2015bandwidth,
	author={Böcherer, Georg and Steiner, Fabian and Schulte, Patrick},
	journal={IEEE Transactions on Communications}, 
	title={Bandwidth Efficient and Rate-Matched Low-Density Parity-Check Coded Modulation}, 
	year={2015},
	volume={63},
	number={12},
	pages={4651-4665},
	doi={10.1109/TCOMM.2015.2494016}}

@ARTICLE{schulte2016constant,
	author={Schulte, Patrick and Böcherer, Georg},
	journal={IEEE Transactions on Information Theory}, 
	title={Constant Composition Distribution Matching}, 
	year={2016},
	volume={62},
	number={1},
	pages={430-434},
	doi={10.1109/TIT.2015.2499181}}

@ARTICLE{ramabadran1990coding,
	author={Ramabadran, T.V.},
	journal={IEEE Transactions on Communications}, 
	title={A coding scheme for m-out-of-n codes}, 
	year={1990},
	volume={38},
	number={8},
	pages={1156-1163},
	doi={10.1109/26.58748}}

@INPROCEEDINGS{kaimalettu2007constellation,
	author={Kaimalettu, Sunil and Thangaraj, Andrew and Bloch, Matthieu and McLaughlin, Steven W.},
	booktitle={2007 IEEE International Symposium on Information Theory}, 
	title={Constellation Shaping using {LDPC} Codes}, 
	year={2007},
	volume={},
	number={},
	pages={2366-2370},
	doi={10.1109/ISIT.2007.4557573}}

@ARTICLE{raphaeli2004constellation,
	author={Raphaeli, D. and Gurevitz, A.},
	journal={IEEE Transactions on Communications}, 
	title={Constellation shaping for pragmatic turbo-coded modulation with high spectral efficiency}, 
	year={2004},
	volume={52},
	number={3},
	pages={341-345},
	doi={10.1109/TCOMM.2004.823564}}

@INPROCEEDINGS{wubben2003mmse,
	author={Wubben, D. and Bohnke, R. and Kuhn, V. and Kammeyer, K.-D.},
	booktitle={VTC 2003-Fall}, 
	title={{MMSE} extension of {V-BLAST} based on sorted {QR} decomposition}, 
	volume={1},
	number={},
	pages={508-512 Vol.1},
	doi={10.1109/VETECF.2003.1285069}}

@ARTICLE{gan2009complex,
	author={Gan, Ying Hung and Ling, Cong and Mow, Wai Ho},
	journal={IEEE Transactions on Signal Processing}, 
	title={Complex Lattice Reduction Algorithm for Low-Complexity Full-Diversity {MIMO} Detection}, 
	year={2009},
	volume={57},
	number={7},
	pages={2701-2710},
	doi={10.1109/TSP.2009.2016267}}

@article{gestner2010lattice,
	title={Lattice reduction for MIMO detection: From theoretical analysis to hardware realization},
	author={Gestner, Brian and Zhang, Wei and Ma, Xiaoli and Anderson, David V},
	journal={IEEE Transactions on Circuits and Systems I: Regular Papers},
	volume={58},
	number={4},
	pages={813--826},
	year={2010},
	publisher={IEEE}
}

@article{pohst1981computation,
	author = {Pohst, Michael},
	title = {On the computation of lattice vectors of minimal length, successive minima and reduced bases with applications},
	year = {1981},
	issue_date = {February 1981},
	publisher = {Association for Computing Machinery},
	address = {New York, NY, USA},
	volume = {15},
	number = {1},
	issn = {0163-5824},
	doi = {10.1145/1089242.1089247},
	journal = {SIGSAM Bull.},
	month = feb,
	pages = {37–44},
	numpages = {8}
}

@ARTICLE{jalden2005complexity,
	author={Jalden, J. and Ottersten, B.},
	journal={IEEE Transactions on Signal Processing}, 
	title={On the complexity of sphere decoding in digital communications}, 
	year={2005},
	volume={53},
	number={4},
	pages={1474-1484},
	doi={10.1109/TSP.2005.843746}}

@article{schnorr1991lattice,
	title={Lattice basis reduction: {I}mproved practical algorithms and solving subset sum problems},
	author={Schnorr, C. P. and Euchner, M.},
	journal={Mathematical Programming},
	year={1991},
	volume={66},
	pages={181-199}
}

@INPROCEEDINGS{yao2002lattice,
	author={Huan Yao and Wornell, G.W.},
	booktitle={Global Telecommunications Conference, 2002. GLOBECOM '02. IEEE}, 
	title={Lattice-reduction-aided detectors for {MIMO} communication systems}, 
	year={2002},
	volume={1},
	number={},
	pages={424-428 vol.1},
	doi={10.1109/GLOCOM.2002.1188114}}

@ARTICLE{studer2008soft,
	author={Studer, Christoph and Burg, Andreas and Bolcskei, Helmut},
	journal={IEEE Journal on Selected Areas in Communications}, 
	title={Soft-output sphere decoding: algorithms and {VLSI} implementation}, 
	year={2008},
	volume={26},
	number={2},
	pages={290-300},
	doi={10.1109/JSAC.2008.080206}}

@inproceedings{bobrov2023probability,
	title={On Probability Shaping for {5G} {MIMO} Wireless Channel with Realistic {LDPC} Codes},
	author={Bobrov, Evgeny and Dordzhiev, Adyan},
	booktitle={International Conference on Mathematical Optimization Theory and Operations Research},
	pages={203--217},
	year={2023},
	organization={Springer}
}

@INPROCEEDINGS{kang2022probabilistic,
	author={Kang, Weimin},
	booktitle={2022 IEEE WCNC}, 
	title={A Probabilistic Shaping Scheme for {MIMO} Systems with Signal Space Diversity}, 
	volume={},
	number={},
	pages={251-255},
	doi={10.1109/WCNC51071.2022.9771617}}

@INPROCEEDINGS{marsh2005smart,
	author={Marsch, P. and Zimmermann, E. and Fettweis, G.},
	booktitle={2005 13th European Signal Processing Conference}, 
	title={Smart candidate adding: A new low-complexity approach towards near-capacity {MIMO} detection}, 
	volume={},
	number={},
	pages={1-4},
	doi={}}

@INPROCEEDINGS{wang2004approaching,
	author={Wang, R. and Giannakis, G.B.},
	booktitle={2004 IEEE WCNC}, 
	title={Approaching {MIMO} channel capacity with reduced-complexity soft sphere decoding}, 
	volume={3},
	number={},
	pages={1620-1625 Vol.3},
	doi={10.1109/WCNC.2004.1311795}}

@standard{IEEE80211n,
	title = {IEEE Standard for Information Technology--Telecommunications and Information Exchange between Systems--Local and Metropolitan Area Networks--Specific Requirements--Part 11: Wireless LAN Medium Access Control (MAC) and Physical Layer (PHY) Specifications},
	author = {{Institute of Electrical and Electronics Engineers}},
	year = {2009},
	number = {802.11n},
	publisher = {IEEE}
}

@INPROCEEDINGS{hocevar2004reduced,
	author={Hocevar, D.E.},
	booktitle={2004 SIPS}, 
	title={A reduced complexity decoder architecture via layered decoding of {LDPC} codes}, 
	pages={107-112},
	doi={10.1109/SIPS.2004.1363033}}

@INPROCEEDINGS{liu2023energy,
	author={Liu, Wei and Richardson, Tom and Shental, Ori and Wu, Liangming and Xu, Changlong and Xu, Hao},
	booktitle={GLOBECOM 2023 - 2023 IEEE Global Communications Conference}, 
	title={Energy-Based Arithmetic Coding Methods for Probabilistic Amplitude Shaping}, 
	year={2023},
	volume={},
	number={},
	pages={5299-5304},
	doi={10.1109/GLOBECOM54140.2023.10436969}
}

@ARTICLE{fehenberger2019multiset,
	author={Fehenberger, Tobias and Millar, David S. and Koike-Akino, Toshiaki and Kojima, Keisuke and Parsons, Kieran},
	journal={IEEE Transactions on Communications}, 
	title={Multiset-Partition Distribution Matching}, 
	year={2019},
	volume={67},
	number={3},
	pages={1885-1893},
	doi={10.1109/TCOMM.2018.2881091}
}

@ARTICLE{richardson2018design,
	author={Richardson, Tom and Kudekar, Shrinivas},
	journal={IEEE Communications Magazine}, 
	title={Design of Low-Density Parity Check Codes for {5G} New Radio}, 
	year={2018},
	volume={56},
	number={3},
	pages={28-34},
	doi={10.1109/MCOM.2018.1700839}}

@INPROCEEDINGS{jia2024simplified,
	author={Jia, Yinhua and Wu, Liangming and Xu, Changlong and Liu, Wei and Xu, Hao and Richardson, Tom},
	booktitle={2024 IEEE International Conference on Communications Workshops (ICC Workshops)}, 
	title={A Simplified Soft-demapping Scheme for Probabilistic Constellation Shaping}, 
	year={2024},
	volume={},
	number={},
	pages={1858-1863},
	doi={10.1109/ICCWorkshops59551.2024.10615331}}

@ARTICLE{gultekin2020enumerative,
	author={Gültekin, Yunus Can and J. van Houtum, Wim and Koppelaar, Arie G. C. and Willems, Frans M. J.},
	journal={IEEE Transactions on Wireless Communications}, 
	title={Enumerative Sphere Shaping for Wireless Communications With Short Packets}, 
	year={2020},
	volume={19},
	number={2},
	pages={1098-1112},
	doi={10.1109/TWC.2019.2951139}}

@standard{DVB-NGH2013,
	title        = {Digital Video Broadcasting (DVB); Next Generation broadcasting system to Handheld, physical layer specification (DVB-NGH)},
	year         = {2013},
	month        = {November},
	note         = {{ETSI EN} 303 105-1}
}

@standard{ATSC-A322,
	title        = {Physical Layer Protocol (A/322)},
	year         = {2017},
	month        = {June},
	note         = {{ATSC A}/322:2017}
}

@ARTICLE{yang2015fifty,
	author={Yang, Shaoshi and Hanzo, Lajos},
	journal={IEEE Communications Surveys and Tutorials}, 
	title={Fifty Years of {MIMO} Detection: The Road to Large-Scale {MIMOs}}, 
	year={2015},
	volume={17},
	number={4},
	pages={1941-1988},
	doi={10.1109/COMST.2015.2475242}}

@ARTICLE{mckay2010achievable,
	author={McKay, Matthew R. and Collings, Iain B. and Tulino, Antonia M.},
	journal={IEEE Transactions on Information Theory}, 
	title={Achievable Sum Rate of {MIMO} {MMSE} Receivers: A General Analytic Framework}, 
	year={2010},
	volume={56},
	number={1},
	pages={396-410},
	doi={10.1109/TIT.2009.2034893}}

@ARTICLE{kim2008performance,
	author={Kim, Namshik and Lee, Yusung and Park, Hyuncheol},
	journal={IEEE Transactions on Wireless Communications}, 
	title={Performance Analysis of {MIMO} System with Linear {MMSE} Receiver}, 
	year={2008},
	volume={7},
	number={11},
	pages={4474-4478},
	doi={10.1109/T-WC.2008.070785}}

@ARTICLE{hu2024supporting,
	author={Hu, Sha and Wang, Hao and Semenov, Sergei},
	journal={IEEE Transactions on Wireless Communications}, 
	title={Supporting Probabilistic Constellation Shaping in {5G-NR} Evolution}, 
	year={2024},
	volume={23},
	number={4},
	pages={3586-3599},
	doi={10.1109/TWC.2023.3309594}}

@INPROCEEDINGS{ivanov2025probabilistic,
	author={Ivanov, Kirill and Yang, Wei and Jiang, Jing},
	booktitle={2025 13th International Symposium on Topics in Coding (ISTC)}, 
	title={Probabilistic Shaping in {MIMO}: Going Beyond 1.53dB {AWGN} Gain With the Non-Linear Demapper}, 
	year={2025},
	volume={},
	number={},
	pages={1-5},
	doi={10.1109/ISTC65386.2025.11154592}}

@ARTICLE{taherzadeh2007lll,
	author={Taherzadeh, Mahmoud and Mobasher, Amin and Khandani, Amir K.},
	journal={IEEE Transactions on Information Theory}, 
	title={{LLL} Reduction Achieves the Receive Diversity in {MIMO} Decoding}, 
	year={2007},
	volume={53},
	number={12},
	pages={4801-4805},
	doi={10.1109/TIT.2007.909169}}

@ARTICLE{campello2018universal,
	author={Campello, Antonio and Ling, Cong and Belfiore, Jean-Claude},
	journal={IEEE Transactions on Information Theory}, 
	title={Universal Lattice Codes for {MIMO} Channels}, 
	year={2018},
	volume={64},
	number={12},
	pages={7847-7865},
	doi={10.1109/TIT.2018.2864745}}
\end{document}